\documentclass[twoside,leqno,twocolumn,11pt]{article}  
\usepackage[margin=1in,dvips]{geometry}
\usepackage{graphicx}
\usepackage{braket}
\usepackage{mathrsfs}
\usepackage{amssymb}
\usepackage{amsthm}
\usepackage{amsmath}

\newtheorem{theorem}{Theorem}[section]
\newtheorem{lemma}{Lemma}[section]
\newtheorem{corollary}{Corollary}[section]
\newtheorem{claim}{Claim}[section]

\newcommand{\abs}[1]{\left|#1\right|}
\newcommand{\floor}[1]{\left\lfloor#1\right\rfloor}
\newcommand{\ceil}[1]{\left\lceil#1\right\rceil}
\newcommand{\norm}[2]{\left \|#2\right \|_{#1}}

\newcommand{\R}{\mathbb{R}}

\newcommand{\Z}{\mathbb{Z}}
\newcommand{\eps}{\epsilon}

\DeclareMathOperator*{\argmin}{arg\,min}

\title{Lower Bounds for Sparse Recovery\thanks{This research has been supported in part by David and Lucille Packard Fellowship,  MADALGO
(Center for Massive Data Algorithmics, funded by the Danish National Research
Association) and NSF grant CCF-0728645.
E. Price has been supported in part by Cisco Fellowship.
}}
\author{Khanh Do Ba\\MIT CSAIL \and Piotr Indyk\\MIT CSAIL \and Eric Price\\MIT CSAIL \and David P. Woodruff\\IBM Almaden}
\date{}
\begin{document}
\maketitle

\begin{abstract}
We consider the following $k$-sparse recovery problem: design an $m \times n$ matrix $A$, such that for any signal $x$, given $Ax$ we can efficiently recover $\hat{x}$ satisfying 
$\norm{1}{x-\hat{x}} \le C \min_{k\mbox{-sparse } x'}  \norm{1}{x-x'}$.
It is known that there exist matrices $A$ with this property that have only $O(k \log (n/k) )$ rows.

In this paper we show that this bound is tight. Our bound holds even for the more general  {\em randomized} version of the problem,  where $A$ is a random variable, and the recovery algorithm is required to work for any fixed $x$ with constant probability (over $A$).
%Here we show an almost tight lower bound of  $m = \Omega(\log n / \log \log n)$ for the case $k=1$.
\end{abstract}

\section{Introduction}

In recent years, a new ``linear'' approach for obtaining a
succinct approximate representation of $n$-dimensional vectors (or
signals) has been discovered.  For any signal $x$, the representation
is equal to $Ax$, where $A$ is an $m \times n$ matrix, or possibly a random
variable chosen from some distribution over such matrices.  The vector $Ax$
is often referred to as the {\em measurement vector} or {\em sketch}
of $x$.  Although $m$ is typically much smaller than $n$, the sketch
$Ax$ contains plenty of useful information about the signal $x$.  A
particularly useful and well-studied problem is that of stable sparse
recovery: given $Ax$, recover a $k$-{\em sparse} vector $\hat{x}$ (i.e.,
having at most $k$ non-zero components) such that
\begin{equation}
\label{e:lplq}
\norm{p}{x-\hat{x}} \le C \min_{k\mbox{-sparse } x'}  \norm{q}{x-x'}
\end{equation}
for some norm parameters $p$ and $q$ and an approximation factor $C=C(k)$. If the matrix $A$ is random, then Equation~\eqref{e:lplq} should hold for each $x$ with some probability (say, 3/4).
 Sparse recovery has  applications to numerous areas such as data stream computing~\cite{Muthu:survey, I-SSS} and compressed sensing~\cite{CRT06:Stable-Signal, Don06:Compressed-Sensing, DDTLTKB}. 

It is known that there exist matrices $A$ and associated recovery algorithms that produce approximations $\hat{x}$ satisfying Equation~\eqref{e:lplq} with $p=q=1$ (i.e., the ``$\ell_1/\ell_1$ guarantee''), constant $C$ and sketch length $m=O(k \log (n/k))$. In particular, a random Gaussian matrix~\cite{CRT06:Stable-Signal}\footnote{In fact, they even achieve a somewhat stronger $\ell_2/\ell_1$ guarantee, see Section~\ref{sec:related}.} or a random sparse binary matrix~(\cite{BGIKS08}, building on \cite{CCF,CM03b})  has this property with overwhelming probability.
 In comparison, using a {\em non-linear} approach, one can obtain a shorter sketch of length  $O(k)$: it suffices to store the  $k$ coefficients with the largest absolute values, together with their indices. 

Surprisingly, it was not known whether the $O( k  \log (n/k))$ bound for linear sketching could be improved upon in general, although such lower bounds were known to hold under certain restrictions (see section~\ref{sec:related} for a more detailed overview).
%although $O(k)$ sketch length was known to suffice if the signal vectors $x$ are required to be {\em exactly} $k$-sparse.
% (and thus the approximation error in Equation~\eqref{e:lplq} is equal to zero).
This  raised hope that the $O(k)$ bound might be achievable even for general vectors $x$. 
Such a scheme would have been of major practical interest, since the sketch length determines the compression ratio, and for large $n$ any extra $\log n$ factor  worsens that ratio tenfold.

%\footnote{It is known that $O(k)$ bound can be achieved under the assumption that the vectors $x$ are $k$-sparse.This represents a major gap in our  understanding of the problem.

In this paper we show that, unfortunately, such an improvement is not possible.
% make significant progress on this issue. 
We address two types of recovery schemes:
\begin{itemize}
\item A {\em deterministic} one, which involves a fixed matrix $A$ and  a recovery algorithm which work for all signals $x$. The aforementioned results  of~\cite{CRT06:Stable-Signal} and others are examples of such schemes.
\item A {\em randomized}  one, where the matrix $A$ is chosen at random from some distribution, and  for each signal $x$ the recovery procedure is correct
%\footnote{Alternatively, one could require the probability of correctness to be only constant. However, all known randomized schemes (such as~\cite{CCF,CM03b}, etc) happen to have the stronger guarantee. } 
with constant probability (say, $3/4$).
Some of the early schemes proposed in the data stream literature (e.g.,~\cite{CCF,CM03b}) belong to this category. 
% \mbox{$1-1/n$} . 
\end{itemize}

Our main result is that, even in the randomized case, the sketch length  $m$ must be at least $\Omega(k \log(n/k))$.
By the aforementioned result of~\cite{CRT06:Stable-Signal} this bound is tight.

%\item Any randomized scheme requires $m=\Omega(\log n  /  \log \log n)$ for $k=1$.  
%\end{itemize}
Thus, our results show that the linear compression is inherently more costly than the simple non-linear approach. 
%However, proving a tight bound for the randomized case remains an open problem.

\subsection{Our techniques}

On a high level, our approach is simple and natural, and utilizes the packing approach: we show that any two ``sufficiently'' different vectors $x$ and $x'$ are mapped to images $Ax$ and $Ax'$ that are ``sufficiently'' different themselves, which requires that the image space is ``sufficiently'' high-dimensional. However, the actual arguments are somewhat subtle.

 Consider first the (simpler) deterministic case. We focus on signals $x=y+z$, where $y$ can be thought of as the ``head'' of the signal and $z$ as the ``tail''. The ``head'' vectors $y$ come from a set $Y$ that is a binary  error-correcting code,  with a minimum distance $\Omega(k)$, where each codeword has weight $k$. On the other hand, the ``tail'' vectors $z$ come from an $\ell_1$ ball (say $B$) with a radius that is a small fraction of $k$. It can be seen that for any two elements $y, y' \in Y$, the balls $y+B$ and $y'+B$, as well as their images, must be disjoint. At the same time, since all vectors $x$ live in a ``large''  $\ell_1$ ball $B'$ of radius $O(k)$, all   images $Ax$ must live in a set $A B'$. The key observation is that the set  $A B'$ is a scaled version of $A(y+B)$  and therefore the ratios of their volumes can be bounded by the scaling factor to the power of the dimension $m$. Since the number of elements of $Y$ is large, this gives a lower bound on $m$.
 
Unfortunately, the aforementioned approach does not seem to extend to the  randomized case.
A natural approach would be to use Yao's principle, and focus on showing a lower bound for a scenario where the matrix $A$ is fixed while the vectors $x=y+z$ are ``random''.
% It might appear that simply choosing vectors $x$ uniformly at random from the set defined above would do the job.
 However, this approach fails, in a very strong sense. Specifically, we are able to show that there is a distribution over matrices $A$ with {\em only} $O(k)$ rows so that for a fixed $y \in Y$ and $z$ chosen uniformly at random from the small ball $B$, we can recover $y$ from $A(y+z)$ with high probability. In a nutshell, the reason is that a random vector from $B$ has an $\ell_2$ norm that is much smaller than the $\ell_2$ norm of elements of $Y$ (even though the $\ell_1$ norms are comparable). This means that the vector $x$ is ``almost'' $k$-sparse (in the $\ell_2$ norm), which enables us to  achieve the $O(k)$ measurement bound.
 
 Instead, we resort to an altogether different approach, via {\em communication complexity}~\cite{kn97}. We start by considering a ``discrete'' scenario where both the matrix $A$ and the vectors $x$ have entries restricted to the polynomial range $\{-n^c \ldots n^c\}$ for some $c=O(1)$. In other words, we assume that the matrix and vector entries can be represented using $O(\log n)$ bits. In this setting we show the following: there is a method for encoding a sequence of $d=O(k \log(n/k) \log n)$ bits into a vector $x$, so that any sparse recovery algorithm can recover that sequence given $Ax$. Since each entry of $Ax$ conveys only $O(\log n)$ bits, it follows that the number $m$ of rows of $A$ must be 
 $\Omega(k \log(n/k))$. 
 
  The encoding is performed by taking
 \[ x=\sum_{j=1}^{\log n} D^j x_j , \]
  where $D = O(1)$ and the $x_j$'s are chosen from the
  error-correcting code $Y$ defined as in the deterministic case.  The
  intuition behind this approach is that a good $\ell_1 / \ell_1$
  approximation to $x$ reveals most of the bits of $x_{\log n}$. This
  enables us to identify $x_{\log n}$ exactly using error
  correction. We could then compute $Ax - A x_{\log n} =
  A(\sum_{j=1}^{\log n-1} D^j x_j )$, and identify $x_{\log n -1}
  \ldots x_1$ in a recursive manner. The only obstacle to completing
  this argument is that we would need the recovery algorithm to work
  for {\em all} $x_i$, which would require lower probability of
  algorithm failure (roughly $1/\log n$). To overcome this problem, we
  replace the encoding argument by a reduction from a related
  communication complexity problem called {\sf Augmented
    Indexing}. This problem has been used in the data stream literature
  \cite{CW09, KNW10} to prove lower bounds for linear algebra
  and norm estimation problems. 
  Since the problem has communication complexity of
  $\Omega(d)$, the conclusion follows.
   
   We apply the argument to arbitrary matrices $A$ by representing
   them as a sum $A'+A''$, where $A'$ has $O(\log n)$ bits of
   precision and $A''$ has ``small'' entries.  We then show that $A'x =
   A(x+s)$ for some $s$ with $\norm{1}{s} <
   n^{-\Omega(1)}\norm{1}{x}$.  In the communication game, this means
   we can transmit $A'x$ and recover $x_{\log n}$ from
   $A'(\sum_{j=1}^{\log n} D^j x_j ) = A(\sum_{j=1}^{\log n} D^j x_j +
   s)$.

   One catch is that $s$ depends on $A$.  The recovery algorithm is
   guaranteed to work with probability $3/4$ for any $x$, so it works
   with probability $3/4$ over any distribution on $x$ independent of
   $A$.  However, there is no guarantee about recovery of $x + s$ when
   $s$ depends on $A$ (even if $s$ is tiny).  To deal with this,
   we choose a $u$ uniformly from the $\ell_1$ ball of radius $k$.  We
   can set $\norm{1}{s} \ll k/n$, so $x + u$ and $x + u + s$ are
   distributions with $o(1)$ statistical distance.  Hence recovery
   from $A(x + u + s)$ matches recovery from $A(x + u)$ with
   probability at least $1 - o(1)$, and $\norm{1}{u}$ is small enough
   that successful recovery from $A(x + u)$ identifies $x_{\log n}$.
   Hence we can recover $x_{\log n}$ from $A(x + u + s) = A'x + Au$
   with probability at least $3/4 - o(1) > 1/2$, which means that the
   {\sf Augmented Indexing} reduction applies to arbitrary matrices as
   well.

 \iffalse
 we focus on signals $x=y+z$ where the ``tail'' vector $z$ is $m$-sparse, and therefore, the $\ell_2$ norms of $y$ and $z$ are not too far from each other. We focus on the case $k=1$ and $m=o(\log n / \log \log n)$. In this case, the sketch $Ay$ of the signal is simply one of the columns of $A$, while $Az$ is a point uniformly distributed in the simplex spanned by the $m$ columns of $A$ corresponding to the support of $z$. We show that, with a probability much higher than $1/n$, the simplex will have volume that is comparable to the volume of the convex hull of all rows in $A$. In that case, we show that the ``noise'' induced by the vector $Az$ makes it impossible to distinguish between $Ax$ and some other column of $A$ with constant probability.
\fi

\subsection{Related Work}
\label{sec:related}

There have been a number of earlier works that have, directly or indirectly,
shown lower bounds for various models of sparse recovery and certain classes
of matrices and algorithms. Specifically, one of the most well-known recovery
algorithms used in compressed sensing is $\ell_1$-minimization, where a
signal $x\in\R^n$ measured by matrix $A$ is reconstructed as
\[
	\hat{x} := \argmin_{x':\,Ax'=Ax} \|x'\|_1.% = x^* + \argmin_{y\in\cN(A)} \|x^*+y\|_1,
\]
%where $\cN(A)$ is the nullspace of $A$.
Kashin and Temlyakov~\cite{KT07} (building on prior work on Gelfand width~\cite{GG,G,Kas}, see also~\cite{Don06:Compressed-Sensing}) gave a characterization of matrices $A$ for which the above recovery algorithm yields the
$\ell_2/\ell_1$ guarantee, i.e.,
\[
	\|x - \hat{x}\|_2 \le C k^{-1/2} \min_{k\mbox{-sparse } x'} \|x - x'\|_1
\]
for some constant $C$, from which it can be shown that such an $A$ must have
$m = \Omega(k\log(n/k))$ rows.

Note that the $\ell_2/\ell_1$ guarantee is somewhat stronger than the $\ell_1/\ell_1$ guarantee investigated in this paper. Specifically, it is easy to observe that if the approximation $\hat{x}$ itself is required to be $O(k)$-sparse, then the $\ell_2/\ell_1$ guarantee implies the $\ell_1/\ell_1$ guarantee (with a somewhat higher approximation constant).
For the sake of simplicity, in this paper we focus mostly on the $\ell_1/\ell_1$ guarantee.
However, our lower bounds apply to the $\ell_2/\ell_1$ guarantee as well: see footnote on page
\pageref{l1l2}.

The results on Gelfand width can be also used to obtain lower bounds for {\em general} recovery algorithms (for the deterministic recovery case), as long as the sparsity parameter $k$ is larger than some constant. This was explicitly stated in~\cite{FPRU}, see also~\cite{Don06:Compressed-Sensing}.

On the other hand, instead of assuming a specific recovery algorithm,
Wainwright~\cite{W07} assumes a specific (randomized) measurement matrix. More
specifically, the author assumes a $k$-sparse
binary signal $x\in\{0,\alpha\}^n$, for some $\alpha>0$, to which is added i.i.d.
standard Gaussian noise in each component. The author then shows that
with a random Gaussian matrix $A$, with each entry also drawn i.i.d. from the
standard Gaussian, we cannot hope to recover $x$ from $Ax$ with any sub-constant
probability of error unless $A$ has $m = \Omega(\frac{1}{\alpha^2} \log \frac{n}{k})$
rows. The author also shows that for $\alpha = \sqrt{1/k}$, this is
tight, i.e., that $m=\Theta(k\log(n/k))$ is both necessary and sufficient. Although
this is only a lower bound for a specific (random) matrix, it is a fairly powerful
one and provides evidence that the often observed upper bound of $O(k\log(n/k))$ is
likely tight.

More recently, Dai and Milenkovic~\cite{DM08}, extending on~\cite{EG88} and~\cite{FR99},
showed an upper bound on superimposed codes that translates to a lower bound on the
number of rows in a compressed sensing matrix that deals only with $k$-sparse
signals but can tolerate measurement noise.
Specifically, if we assume a $k$-sparse signal $x\in([-t,t]\cap\Z)^n$,
and that arbitrary noise $\mu\in\R^n$ with $\|\mu\|_1 < d$ is
added to the measurement vector $Ax$, then if exact recovery is still possible, $A$
must have had $m \ge C k\log n / \log k$ rows, for some constant $C = C(t,d)$ and
sufficiently large $n$ and $k$.\footnote{Here $A$ is assumed to
have its columns normalized to have $\ell_1$-norm 1. This is natural since
otherwise we could simply scale $A$ up to make the image points $Ax$ arbitrarily
far apart, effectively nullifying the noise.}

\section{Preliminaries}

In this paper we focus on recovering sparse approximations $\hat{x}$ that satisfy the following $C$-approximate $\ell_1/\ell_1$ guarantee with sparsity parameter $k$:

\begin{equation}
\label{e:l1l1}
\norm{1}{x-\hat{x}} \leq C \min_{k\mbox{-sparse } x'}\norm{1}{x-x'}.
\end{equation}

We define a $C$-approximate {\em deterministic} $\ell_1/\ell_1$ recovery algorithm  to be a pair $(A, \mathscr{A})$ where $A$ is an $m\times
n$ observation matrix and $\mathscr{A}$ is an algorithm that, for any $x$, maps
$Ax$ (called the {\em sketch} of $x$) to some $\hat{x}$ that satisfies Equation~\eqref{e:l1l1}.
% the  $C$-approximate $\ell_1/\ell_1$ guarantee with sparsity parameter $k$.
%\[
%\norm{1}{x-x^*} \leq C \min_{k-\mbox{sparse } x'}\norm{1}{x-x'}.
%\]

We define a $C$-approximate {\em randomized}  $\ell_1/\ell_1$  recovery algorithm to be a pair $(A, \mathscr{A})$ where $A$ is a {\em random variable} chosen from some distribution over $m\times n$ measurement matrices, and $\mathscr{A}$ is an algorithm which, for any $x$,  maps a pair $(A, Ax)$ to some $\hat{x}$ that satisfies Equation~\eqref{e:l1l1}
%the $C$-approximate $\ell_1/\ell_1$ guarantee with sparsity parameter $k$ 
%\[
%\norm{1}{x-x^*} \leq C \min_{k-\mbox{sparse } x'}\norm{1}{x-x'}.
%\] 
with probability at least $3/4$.

We use $B^n_p(r)$ to denote the $\ell_p$ ball of radius $r$ in $\R^n$; we skip the superscript $n$ if it is clear from the context.

For any vector $x$, we use $\|x\|_0$ to denote the ``$\ell_0$ norm of $x$'', i.e., the number of non-zero entries in $x$.

\section{Deterministic Lower Bound}
\label{sec:detlb}

We will prove a lower bound on $m$ for any $C$-approximate
deterministic recovery algorithm.  First we use a discrete volume
bound (Lemma~\ref{g-v}) to find a large set $Y$ of points that are at
least $k$ apart from each other.  Then we use another volume bound
(Lemma~\ref{volumes}) on the images of small $\ell_1$ balls around
each point in $Y$.  If $m$ is too small, some two images collide.  But
the recovery algorithm, applied to a point in the collision, must
yield an answer close to two points in $Y$.  This is impossible, so
$m$ must be large.

\begin{lemma}(Gilbert-Varshamov)\label{g-v}
  For any $q, k \in \mathbb{Z}^+, \epsilon \in \mathbb{R}^+$ with
  $\epsilon < 1-1/q$, there exists a set $Y \subset \{0,1\}^{qk}$ of
  binary vectors with exactly $k$ ones, such that $Y$ has minimum
  Hamming distance $2\epsilon k$ and
  \[
  \log \abs{Y} > (1 - H_{q}(\epsilon))k \log q
  \]
  where $H_q$ is the $q$-ary entropy function $H_q(x) =
  -x\log_q\frac{x}{q-1} - (1-x)\log_q (1-x)$.
\end{lemma}

See appendix for proof.

\begin{lemma}\label{volumes}
 Take an $m\times n$ real matrix $A$, positive reals $\epsilon, p,
 \lambda$, and $Y \subset B_p^n(\lambda)$.  If $\abs{Y} > (1 +
 1/\epsilon)^m$, then there exist $z, \overline{z} \in
 B_p^n(\eps\lambda)$ and $y, \overline{y} \in Y$ with $y \neq
 \overline{y}$ and $A(y+z) = A(\overline{y}+\overline{z})$.
\end{lemma}
\begin{proof}
  If the statement is false, then the images of all $\abs{Y}$ balls $\{y +
  B_p^n(\epsilon\lambda) \mid y \in Y\}$ are disjoint.  However, those
  balls all lie within $B_p^n((1+\epsilon)\lambda)$, by the bound on the norm
  of $Y$.  A volume argument gives the result, as follows.

  Let $S = AB_p^n(1)$ be the image of the $n$-dimensional ball of
  radius $1$ in $m$-dimensional space.  This is a polytope with some
  volume $V$.  The image of $B_p^n(\epsilon\lambda)$ is a linearly
  scaled $S$ with volume $(\epsilon\lambda)^mV$, and the volume of the
  image of $B_p^n((1+\epsilon)\lambda)$ is similar with volume
  $((1+\epsilon)\lambda)^mV$.  If the images of the former are all
  disjoint and lie inside the latter, we have
  $\abs{Y}(\epsilon\lambda)^mV \leq ((1+\epsilon)\lambda)^mV$,
  or $\abs{Y} \leq (1+1/\epsilon)^m$.  If $Y$ has more elements than
  this, the images of some two balls $y + B_p^n(\epsilon\lambda)$ and
  $\overline{y} + B_p^n(\eps\lambda)$ must intersect, implying the lemma.
\end{proof}

\begin{theorem}
  Any $C$-approximate deterministic recovery algorithm must have
  \[
  m \geq \frac{1 - H_{\floor{n/k}}(1/2)}{\log (4 + 2C)} k \log\floor{\frac{n}{k}}.
  \]
\end{theorem}

\begin{proof}
  Let $Y$ be a maximal set of $k$-sparse $n$-dimensional binary
  vectors with minimum Hamming distance $k$, and let $\gamma =
  \frac{1}{3 + 2C}$.  By Lemma~\ref{g-v} with $q = \floor{n/k}$
  we have $\log \abs{Y} > (1 -
  H_{\floor{n/k}}(1/2))k\log{\floor{n/k}}$.

  Suppose that the theorem is not true; then $m < \log \abs{Y} / \log
  (4+2C) = \log \abs{Y} / \log (1 + 1/\gamma)$, or $\abs{Y} > (1 +
  \frac{1}{\gamma})^m$.  Hence Lemma~\ref{volumes} gives us
  some $y, \overline{y} \in Y$ and $z, \overline{z}\in B_1(\gamma k)$ with
  $A(y + z) = A(\overline{y}+\overline{z})$.

  Let $w$ be the result of running the recovery algorithm on $A(y+z)$.  By the
  definition of a deterministic recovery algorithm, we have
  \begin{align*}
  &\norm{1}{y+z-w} \leq C \min_{k\mbox{-sparse } y'} \norm{1}{y+z-y'}\\
  &\norm{1}{y-w} - \norm{1}{z} \leq C \norm{1}{z}\\
  &\norm{1}{y-w}   \leq (1+C) \norm{1}{z}
  								 \leq (1+C)\gamma k
  								 = \tfrac{1+C}{3+2C}\,k,
  \end{align*}
  and similarly $\norm{1}{\overline{y}-w} \leq \frac{1+C}{3+2C}\,k$, so
  \begin{align*}
    \norm{1}{y-\overline{y}} &\leq \norm{1}{y-w} + \norm{1}{\overline{y}-w}
    = \frac{2+2C}{3+2C}k < k.
  \end{align*}
  But this contradicts the definition of $Y$, so $m$ must be large
  enough for the guarantee to hold.
\end{proof}

\begin{corollary}
  If $C$ is a constant bounded away from zero, then $m = \Omega(k\log(n/k))$.
\end{corollary}

\section{Randomized Upper Bound for Uniform Noise}\label{upper-bound}

The standard way to prove a randomized lower bound is to find a
distribution of hard inputs, and to show that any deterministic
algorithm is likely to fail on that distribution.  In our context, we
would like to define a ``head'' random variable $y$ from a distribution $Y$  and a
``tail'' random variable $z$ from a distribution $Z$, such that any algorithm given the sketch of $y+z$ must recover an incorrect $y$ with non-negligible probability.

Using our deterministic bound as inspiration, we could take $Y$ to be
uniform over a set of $k$-sparse binary vectors of minimum Hamming
distance $k$ and $Z$ to be uniform over the ball $B_1(\gamma k)$ for some constant $\gamma>0$.  Unfortunately, as the following theorem shows,
%Section~\ref{upper-bound},
one can actually perform a recovery of such vectors using only $O(k)$ measurements; this is because $\norm{2}{z}$ is very small (namely, $\tilde{O}(k / \sqrt{n})$) with
high probability.

%The first and most natural attempt to generalize our deterministic
%lower bound to a randomized one is to take the same sets of signals
%and noise, and introduce randomness by imposing the uniform
%distribution on them. The hope is then that this is a ``hard
%distribution'' for any deterministic algorithm, i.e., that any fixed
%matrix with too few rows must confuse at least two of the signals with
%non-negligible probability, so that choosing a random signal from our
%space of signals would result in the wrong decoding more than a $1/n$
%fraction of the time. However, in this section we show that this is
%not possible by giving a $Y$-recovery algorithm that solves this
%special case with only $O(k)$ measurements.
%
%The exact statement is as follows.

\begin{theorem}
\label{thm:unifub}
	Let $Y \subset \R^n$ be a set of signals with the property that for every distinct
	$y_1,y_2\in Y$, $\|y_1-y_2\|_2 \ge r$, for some parameter $r>0$.
	Consider ``noisy signals'' $x=y+z$, where $y \in Y$ and $z$
	is a  ``noise vector'' chosen uniformly at random from $B_1(s)$, for another
	parameter $s>0$. Then using an $m\times n$ Gaussian measurement matrix
	$A = (1/\sqrt{m})(g_{ij})$, where $g_{ij}$'s are i.i.d. standard
	Gaussians, we can recover $y\in Y$ from $A(y+z)$ with probability
	$1-1/n$ (where the probability is over both $A$ and $z$), as long as
	\[
		s \le O\left(\frac{r m^{1/2} n^{1/2-1/m}}{|Y|^{1/m} \log^{3/2}n}\right).
	\]
\end{theorem}

To prove the theorem we will need the following two lemmas.

\begin{lemma}
\label{lem:l2pres}
	For any $\delta>0$, $y_1,y_2\in Y$, $y_1\not=y_2$, and $z\in \R^n$, each of
	the following holds with probability at least $1-\delta$:
	\begin{itemize}
		\item $\|A(y_1-y_2)\|_2 \ge \frac{\delta^{1/m}}{3}\|y_1-y_2\|_2$, and
		\item $\|Az\|_2 \le (\sqrt{(8/m)\log(1/\delta)}+1)\|z\|_2$.
	\end{itemize}
\end{lemma}

See the appendix for the proof.

\begin{lemma}\label{lem:l2small}
	A random vector $z$ chosen uniformly from $B_1(s)$ satisfies
	\[
		\Pr[\|z\|_2 > \alpha s\log n/\sqrt{n}] < 1/n^{\alpha-1}.
	\]
\end{lemma}

See the appendix for the proof.\\

\noindent{\it Proof of theorem.}\;\;
	In words, Lemma \ref{lem:l2pres} says that $A$ cannot bring faraway signal points
	too close together, and cannot blow up a small noise vector too much. Now, we
	already assumed the signals to be far apart, and Lemma \ref{lem:l2small} tells us
	that the noise is indeed small (in $\ell_2$ distance). The result is that in the
	image space, the noise is not enough to confuse different signals.
	Quantitatively, applying the second part of Lemma \ref{lem:l2pres} with
	$\delta = 1/n^2$, and Lemma \ref{lem:l2small} with $\alpha=3$, gives us
	\begin{equation}\label{eq:noise}
		\|Az\|_2
			\le O\left(\frac{\log^{1/2}n}{m^{1/2}}\right)\|z\|_2
			\le O\left(\frac{s\log^{3/2}n}{(m n)^{1/2}}\right)
	\end{equation}
	with probability $\ge 1-2/n^2$. On the other hand, given signal $y_1\in Y$, we know
	that every other signal $y_2\in Y$ satisfies $\|y_1-y_2\|_2 \ge r$, so by the
	first part of Lemma \ref{lem:l2pres} with $\delta = 1/(2n|Y|)$, together with a
	union bound over every $y_2\in Y$,
	\begin{equation}\label{eq:sigs}
		\|A(y_1-y_2)\|_2
			\ge \frac{\|y_1-y_2\|_2}{3(2n|Y|)^{1/m}}
			\ge \frac{r}{3(2n|Y|)^{1/m}}
	\end{equation}
	holds for every $y_2\in Y$, $y_2\not=y_1$, simultaneously with probability
	$1-1/(2n)$.

	Finally, observe that as long as $\|Az\|_2 < \|A(y_1-y_2)\|_2/2$ for every
	competing signal $y_2\in Y$, we are guaranteed that
	\begin{eqnarray*}
		\|A(y_1+z) - Ay_1\|_2
			&=& \|Az\|_2\\
%			< \|A(y_1-y_2)\|_2/2
			&<& \|A(y_1-y_2)\|_2 - \|Az\|_2\\
			&\le& \|A(y_1+z) - Ay_2\|_2
	\end{eqnarray*}
	for every $y_2\not=y_1$, so we can recover $y_1$ by simply returning the signal
	whose image is closest to our measurement point $A(y_1+z)$ in $\ell_2$ distance.
	To achieve this, we can chain Equations (\ref{eq:noise}) and (\ref{eq:sigs})
	together (with a factor of 2), to see that
	\[
		s \le O\left(\frac{r m^{1/2} n^{1/2-1/m}}{|Y|^{1/m} \log^{3/2}n}\right)
	\]
	suffices. Our total probability of failure is at most $2/n^2 + 1/(2n) < 1/n$.
\\

The main consequence of this theorem is that for the setup we used in Section
\ref{sec:detlb} to prove a deterministic lower bound of $\Omega(k\log(n/k))$,
if we simply draw the noise uniformly randomly from the same $\ell_1$ ball (in fact,
even one with a much larger radius, namely, polynomial in $n$), this ``hard
distribution'' can be defeated with just $O(k)$ measurements:

\begin{corollary}
	If $Y$ is a set of binary $k$-sparse vectors, as in Section~\ref{sec:detlb}, and noise
	$z$ is drawn uniformly at random from $B_1(s)$, then for any constant $\eps>0$, $m = O(k/\eps)$
	measurements suffice to recover any signal in $Y$ with probability $1-1/n$, as long as
	\[
		s \le O\left(\frac{k^{3/2+\eps} n^{1/2-\eps}}{\log^{3/2}n}\right).
	\]
\end{corollary}

\begin{proof}
	The parameters in this case are $r=k$ and $|Y| \le \binom{n}{k} \le (ne/k)^k$, so by Theorem
	\ref{thm:unifub}, it suffices to have
	\[
		s \le O\left(\frac{k^{3/2+k/m} n^{1/2-(k+1)/m}}{\log^{3/2}n}\right).
	\]
	Choosing $m = (k+1)/\eps$ yields the corollary.
\end{proof}

\section{Randomized Lower Bound}

Although it is possible to partially circumvent this obstacle by focusing our noise
distribution on ``high'' $\ell_2$ norm, sparse vectors, we are able to obtain
stronger results via a reduction from a communication game and the corresponding
lower bound.

The communication game will show that a message $Ax$ must have a large
number of bits.  To show that this implies a lower bound on the number
of rows of $A$, we will need $A$ to be discrete.  Hence we first show
that discretizing $A$ does not change its recovery characteristics by
much.

\subsection{Discretizing Matrices}
Before we discretize by rounding, we need to ensure that the matrix is
well conditioned.  We show that without loss of generality, the rows
of $A$ are orthonormal.

We can multiply $A$ on the left by any invertible matrix to get
another measurement matrix with the same recovery characteristics.  If
we consider the singular value decomposition $A = U\Sigma V^*$, where
$U$ and $V$ are orthonormal and $\Sigma$ is 0 off the diagonal, this
means that we can eliminate $U$ and make the entries of $\Sigma$ be
either $0$ or $1$.  The result is a matrix consisting of $m$
orthonormal rows.  For such matrices, we prove the following:

\begin{lemma}\label{thm:discretizing}
  Consider any $m\times n$ matrix $A$ with orthonormal rows.  Let $A'$
  be the result of rounding $A$ to $b$ bits per entry.  Then for any
  $v \in \mathbb{R}^n$ there exists an $s \in \mathbb{R}^n$ with $A'v
  = A(v - s)$ and $\norm{1}{s} < n^22^{-b}\norm{1}{v}$.
\end{lemma}
\begin{proof}
  Let $A'' = A - A'$ be the roundoff error when discretizing $A$ to
  $b$ bits, so each entry of $A''$ is less than $2^{-b}$.  Then for
  any $v$ and $s = A^TA''v$, we have $As = A''v$ and
  \begin{align*}
    \norm{1}{s} &= \norm{1}{A^TA''v} \leq \sqrt{n}\norm{1}{A''v}\\ &\leq
    m\sqrt{n}2^{-b}\norm{1}{v} \leq n^22^{-b}\norm{1}{v}.
  \end{align*}
\end{proof}

\subsection{Communication Complexity}

We use a few definitions and results from two-party communication complexity. For further background see the book by Kushilevitz and Nisan \cite{kn97}. Consider the following communication game. There are two parties, Alice and Bob. Alice is given a string $y \in \{0,1\}^d$. Bob is given an index $i \in [d]$, together with $y_{i+1}, y_{i+2}, \ldots, y_d$. The parties also share an arbitrarily long common random string $r$. Alice sends a single message $M(y,r)$ to Bob, who must output $y_i$ with probability at least $2/3$, where the probability is taken over $r$. We refer to this problem as {\sf Augmented Indexing}. The communication cost of {\sf Augmented Indexing} is the minimum, over all correct protocols, of the length of the message $M(y,r)$ on the worst-case choice of $r$ and $y$.

The next theorem is well-known and follows from Lemma 13 of \cite{MNSW98} (see also Lemma 2 of \cite{BJKK04}). 
\begin{theorem}\label{thm:AIND}
The communication cost of {\sf Augmented Indexing} is $\Omega(d)$.
\end{theorem}
\begin{proof}
First, consider the private-coin version of the problem, in which both parties can toss coins, but do not share a random string $r$ (i.e., there is no public coin). Consider any correct protocol for this problem. We can assume the probability of error of the protocol is an arbitrarily small positive constant by increasing the length of Alice's message by a constant factor (e.g., by independent repetition and a majority vote). Applying Lemma 13 of \cite{MNSW98} (with, in their notation, $t = 1$ and $a = c' \cdot d$ for a sufficiently small constant $c' > 0$), the communication cost of such a protocol must be $\Omega(d)$. Indeed, otherwise there would be a protocol in which Bob could output $y_i$ with probability greater than $1/2$ without any interaction with Alice, contradicting that $\Pr[y_i = 1/2]$ and that Bob has no information about $y_i$. Our theorem now follows from Newman's theorem (see, e.g., Theorem 2.4 of \cite{knr99}), which shows that the communication cost of the best public coin protocol is at least that of the private coin protocol minus $O(\log d)$ (which also holds for one-round protocols). 
\end{proof}
%Using this theorem we can prove a lower bound on the number of rows of a matrix $A$ used for compressed sensing, provided the entries are polynomially bounded. 

%Let $\mathcal{F}$ denote the family of all matrices with $n$ columns and integer entries bounded in absolute value by $\poly(n)$. For $x \in \mathbb{R}^n$ and $A \in \mathcal{F}$, we refer to $Ax$ as the {\it sketch} of $x$. For a compressed sensing scheme, we let $E$ denote the estimation algorithm, that is, the algorithm that takes in $A$ together with $Ax$ and outputs a $k$-sparse vector $x'$. Let $\nu$ be any distribution on $\mathcal{F}$.

\subsection{Randomized Lower Bound Theorem}

\begin{theorem}\label{thm:discrete-randomized}
For any randomized $\ell_1/\ell_1$ recovery algorithm
$(A,\mathscr{A})$, with approximation factor $C = O(1)$, $A$ must have
$m = \Omega(k\log(n/k))$ rows.
\end{theorem}
\begin{proof}
We shall assume, without loss of generality, that $n$ and $k$ are
powers of $2$, that $k$ divides $n$, and that the rows of $A$ are
orthonormal.  The proof for the general case follows with minor
modifications.

Let $(A,\mathscr{A})$ be such a recovery algorithm. We will show how
to solve the {\sf Augmented Indexing} problem on instances of size $d
= \Omega(k \log (n/k) \log n)$ with communication cost $O(m \log
n)$. The theorem will then follow by Theorem \ref{thm:AIND}.

Let $X$ be the maximal set of $k$-sparse $n$-dimensional binary
vectors with minimum Hamming distance $k$.  From Lemma~\ref{g-v} we
have $\log \abs{X} = \Omega(k \log (n/k))$.  Let $d = \floor{\log
  \abs{X}} \log n$, and define $D = 2C+3$.

Alice is given a string $y \in \{0,1\}^d$, and Bob is given $i \in
[d]$ together with $y_{i+1}, y_{i+2}, \ldots, y_d$, as in the setup
for {\sf Augmented Indexing}.

Alice splits her string $y$ into $\log n$ contiguous chunks $y^1, y^2,
\ldots, y^{\log n}$, each containing $\floor{\log \abs{X}}$ bits. She
uses $y^j$ as an index into $X$ to choose $x_j$.  Alice defines
\[
x = D^1 x_1 + D^2 x_2 + \cdots + D^{\log n} x_{\log n}.
\]
Alice and Bob use the common randomness $r$ to agree upon a random
matrix $A$ with orthonormal rows.  Both Alice and Bob round $A$ to form
$A'$ with $b = \ceil{(4+2\log D)\log n} = O(\log n)$ bits per entry.
Alice computes $A'x$ and transmits it to Bob.

From Bob's input $i$, he can compute the value $j=j(i)$ for which the
bit $y_i$ occurs in $y^{j}$. Bob's input also contains $y_{i+1},
\ldots, y_n$, from which he can reconstruct $x_{j+1}, \ldots, x_{\log
  n}$, and in particular can compute
\[
z = D^{j+1}x_{j+1} + D^{j+2}x_{j+2} + \cdots + D^{\log n}x_{\log n}.
\]
Set $w = x-z = \sum_{i=1}^j D^ix_i$.  Bob then computes $A'z$, and
using $A'x$ and linearity, $A'w$.  Then
\[
\norm{1}{w} \leq \sum_{i=1}^j kD^i < k\frac{D^{1 + j}}{D-1} < kD^{2\log n}.
\]
So from Lemma~\ref{thm:discretizing}, there exists some $s$ with $A'w
= A(w-s)$ and
\[
\norm{1}{s} < n^22^{-3\log n - 2\log D \log n}\norm{1}{w} < k/n^2.
\]
Bob chooses another vector $u$ uniformly from $B_1^n(k)$, the $\ell_1$
ball of radius $k$, and computes $A(w - s - u) = A'w - Au$.

Bob runs the estimation algorithm $\mathscr{A}$ on $A$ and $A(w-s-u)$,
obtaining $\hat{w}$.  We have that $u$ is independent of $w$ and $s$,
and that $\norm{1}{u} \leq k(1 - 1/n^2) \leq k - \norm{1}{s}$ with
probability $\frac{\text{Vol}(B_1^n(k(1 - 1/n^2)))}{\text{Vol}(B_1^n(k))} = (1 - 1/n^2)^n > 1 - 1/n$.  But $\{w - u \mid \norm{1}{u}
\leq k - \norm{1}{s}\} \subseteq \{w - s - u \mid \norm{1}{u} \leq
k\}$, so 
%as a distribution over $u$, 
the ranges of the random variables $w - s - u$ and $w
- u$ overlap in at least a $1 - 1/n$ fraction of their
volumes. Therefore $w-s-u$ and $w-u$ have statistical distance at most
$1/n$.  The distribution of $w - u$ is independent of $A$, so running
the recovery algorithm on $A(w - u)$ would work with probability at
least $3/4$.  Hence with probability at least $3/4 - 1/n\ge 2/3$ (for $n$ large enough), $\hat{w}$
satisfies the recovery criterion for $w - u$, meaning
\[
\norm{1}{w-u-\hat{w}} \leq C \min_{k\mbox{-sparse } w'} \norm{1}{w-u-w'}.
\]
Now,
\begin{align*}
  \norm{1}{D^jx_j - \hat{w}} &\leq  \norm{1}{w-u - D^jx_j} +  \norm{1}{w-u - \hat{w}}\\
  &\leq (1 + C) \norm{1}{w-u - D^jx_j}\\
  &\leq (1 + C)(\norm{1}{u} + \sum_{i=1}^{j-1} \norm{1}{D^ix_i})\\
  &\leq (1 + C)k\sum_{i=0}^{j-1} D^{i}\\
  &< k \cdot \frac{(1+C)D^{j}}{D-1}\\
  &= kD^j/2.
\end{align*}
And since the minimum Hamming distance in $X$ is $k$, this means
$\norm{1}{D^jx_j - \hat{w}} < \norm{1}{D^jx' - \hat{w}}$ for all $x'
\in X, x' \neq x_j$\footnote{Note that these bounds would still hold
  with minor modification if we replaced the $\ell_1/\ell_1$
  guarantee with the $\ell_2/\ell_1$ guarantee, so the same result
  holds in that case.\label{l1l2}}.  So Bob can correctly identify
$x_j$ with probability at least $2/3$.  From $x_j$ he can recover
$y^j$, and hence the bit $y_i$ that occurs in $y^j$.

Hence, Bob solves {\sf Augmented Indexing} with probability at least
$2/3$ given the message $A'x$.  The entries in $A'$ and $x$ are
polynomially bounded integers (up to scaling of $A'$), and so each
entry of $A'x$ takes $O(\log n)$ bits to describe. Hence, the
communication cost of this protocol is $O(m \log n)$. By Theorem
\ref{thm:AIND}, $m \log n = \Omega(k \log(n/k) \log n)$, or $m =
\Omega(k \log (n/k))$.
\end{proof}

\bibliographystyle{alpha}
\bibliography{compressed}

\iffalse

\fi

\appendix

\section{ Proof of Lemma~\ref{g-v} }

%\label{g-v}
\begin{proof}
  We will construct a codebook $T$ of block length $k$, alphabet $q$,
  and minimum Hamming distance $\epsilon k$.  Replacing each character
  $i$ with the $q$-long standard basis vector $e_i$ will create a binary
  $qk$-dimensional codebook $S$ with minimum Hamming distance $2
  \epsilon k$ of the same size as $T$, where each element of $S$ has
  exactly $k$ ones.

  The Gilbert-Varshamov bound, based on volumes of Hamming balls,
  states that a codebook of size $L$ exists for some
  \[
  L \geq \frac{q^k}{\sum_{i=0}^{\epsilon k - 1} \binom{k}{i} (q-1)^i}.
  \]
  Using the claim (analogous to~\cite{vL98}, p. 21, proven
  below) that for $\epsilon < 1 - 1/q$
  \[
  \sum_{i=0}^{\epsilon k}\binom{k}{i} (q-1)^i < q^{H_q(\epsilon)k},
  \]
  we have that $\log L > (1 - H_q(\epsilon))k \log q$, as desired.
\end{proof}

\begin{claim}
  For $0 < \epsilon < 1 - 1/q$,
  \[
  \sum_{i=0}^{\epsilon k}\binom{k}{i} (q-1)^i < q^{H_q(\epsilon)k}.
  \]
\end{claim}

\begin{proof}
  Note that
  \[
  q^{-H_q(\epsilon)} = \left(\frac{\epsilon}{(q-1)(1-\epsilon)}\right)^\epsilon(1-\epsilon) < (1 - \epsilon).
  \]
  Then
  \begin{align*}
    1 &= (\epsilon + (1 - \epsilon))^k\\
    &> \sum_{i=0}^{\epsilon k} \binom{k}{i} \epsilon^i(1-\epsilon)^{k-i}\\
    &= \sum_{i=0}^{\epsilon k} \binom{k}{i}(q-1)^i \left(\frac{\epsilon}{(q-1)(1-\epsilon)}\right)^i(1-\epsilon)^{k}\\
    &> \sum_{i=0}^{\epsilon k} \binom{k}{i}(q-1)^i \left(\frac{\epsilon}{(q-1)(1-\epsilon)}\right)^{\epsilon k}(1-\epsilon)^{k}\\
    &= q^{-H_q(\epsilon)k}\sum_{i=0}^{\epsilon k} \binom{k}{i}(q-1)^i
  \end{align*}
\end{proof}

\section{ Proof of Lemma~\ref{lem:l2pres} }

%\label{lem:l2pres}
\begin{proof}
	By standard arguments (see, e.g.,~\cite{IN07}), for any $D>0$ we have
	\[
		\Pr\left[\|A(y_1-y_2)\|_2 \le \frac{\|y_1-y_2\|_2}{D}\right]
			\le \left(\frac{3}{D}\right)^m
	\]
	and
	\[
		\Pr[\|Az\|_2 \ge D\|z\|_2] \le e^{-m(D-1)^2/8}.
	\]
	Setting both right-hand sides to $\delta$ yields the lemma.
\end{proof}

\section{Proof of Lemma~\ref{lem:l2small} }

%\label{lem:l2small}
\begin{proof}
	Consider the distribution of a single coordinate of $z$, say, $z_1$. The
	probability density of $|z_1|$ taking value $t\in[0,s]$ is proportional to the
	$(n-1)$-dimensional volume of $B_1^{(n-1)}(s-t)$, which in turn is proportional
	to $(s-t)^{n-1}$. Normalizing to ensure the probability integrates to 1, we
	derive this probability as
	\[
		p(|z_1|=t) = \frac{n}{s^n}(s-t)^{n-1}.
	\]
	It follows that, for any $D\in[0,s]$,
	\[
		\Pr[|z_1| > D] = \int_D^s \frac{n}{s^n}(s-t)^{n-1}\;dt= (1-D/s)^n.
	\]
	In particular, for any $\alpha>1$,
	\begin{eqnarray*}
		\Pr[|z_1| &>& \alpha s\log n/n] = (1 - \alpha\log n/n)^n < e^{-\alpha\log n}\\
							&=& 1/n^\alpha.
	\end{eqnarray*}
	Now, by symmetry this holds for every other coordinate $z_i$ of $z$ as well, so by
	the union bound
	\[
		\Pr[\|z\|_\infty > \alpha s\log n/n] < 1/n^{\alpha-1},
	\]
	and since $\|z\|_2 \le \sqrt{n}\cdot\|z\|_\infty$ for any vector $z$, the lemma
	follows.
\end{proof}

\end{document}